\documentclass[aip,cha,showpacs,10pt,twocolumn,superscriptaddress]{revtex4-1}

\usepackage{amsmath,amsthm}
\usepackage{amssymb}
\usepackage{amscd}
\usepackage{textcomp}
\usepackage{fancyvrb}

\usepackage{wasysym}
\usepackage[ansinew]{inputenc}
\usepackage[T1]{fontenc}
\usepackage{ae,aecompl}
\usepackage{hyperref}
\usepackage{sidecap}

\usepackage[pdftex]{graphicx}
\DeclareGraphicsExtensions{.pdf}

\usepackage{listings}    

\usepackage{color}
\definecolor{red}{rgb}{1,0,0}
\definecolor{blue}{rgb}{0,0,1}
\definecolor{green}{rgb}{0,0.5,0}
\definecolor{magenta}{rgb}{1,0,1}

\newcommand{\be}{\begin{equation}}
\newcommand{\ee}{\end{equation}}
\newcommand{\bea}{\begin{eqnarray}}
\newcommand{\eea}{\end{eqnarray}}
\newcommand{\bal}{\begin{align}}
\newcommand{\eal}{\end{align}}
\newcommand{\nn}{\nonumber}
\newcommand{\eye}{\mbox{$\mbox{1}\!\mbox{l}\;$}}

\renewcommand{\vec}[1]{\boldsymbol{#1}}

\newtheorem{thm}{Theorem}
\newtheorem{defn}[thm]{Definition}
\newtheorem{prop}{Proposition}

\newtheorem{lemma}{Lemma}

\newtheorem{corr}{Corollary}

\begin{document}

\title{Rotor-angle versus voltage instability in the third-order model for synchronous generators}

\author{Konstantin~Sharafutdinov}
\email[Electronic mail: ]{konstantin.sharafutdinov@phystech.edu}
\affiliation{Forschungszentrum J\"ulich, Institute for Energy and Climate Research -
Systems Analysis and Technology Evaluation (IEK-STE), 52428 J\"ulich, Germany}
\affiliation{Institute for Theoretical Physics, University of Cologne, 50937 K\"oln, Germany}

\author{Leonardo~Rydin~Gorj\~ao}
\email[Electronic mail: ]{leonardo.rydin@gmail.com}
\affiliation{Forschungszentrum J\"ulich, Institute for Energy and Climate Research -
Systems Analysis and Technology Evaluation (IEK-STE), 52428 J\"ulich, Germany}

\author{Moritz~Matthiae}
\email[Electronic mail: ]{momat@nanotech.dtu.dk}
\affiliation{Forschungszentrum J\"ulich, Institute for Energy and Climate Research -
Systems Analysis and Technology Evaluation (IEK-STE), 52428 J\"ulich, Germany}
\affiliation{Department of Micro- and Nanotechnology, Technical University of Denmark, 2800 Kongens Lyngby, Denmark}

\author{Timm~Faulwasser}
\email[Electronic mail: ]{timm.faulwasser@kit.edu}
\affiliation{Institute for Automation and Applied Informatics, Karlsruhe Institute of Technology, 76344 Karlsruhe, Germany}

\author{Dirk~Witthaut}
\email[Electronic mail: ]{d.witthaut@fz-juelich.de}
\affiliation{Forschungszentrum J\"ulich, Institute for Energy and Climate Research -
Systems Analysis and Technology Evaluation (IEK-STE), 52428 J\"ulich, Germany}
\affiliation{Institute for Theoretical Physics, University of Cologne, 50937 K\"oln, Germany}



\markboth{Journal of \LaTeX\ Class Files,~Vol.~13, No.~9, September~2014}%
{Shell \MakeLowercase{\textit{et al.}}: Bare Demo of IEEEtran.cls for Journals}

\begin{abstract}
We investigate the interplay of rotor-angle and voltage stability in electric power systems. To this end, we carry out a local stability analysis of the third-order model which entails the classical power-swing equations and the voltage dynamics. We provide necessary and sufficient stability conditions and investigate different routes to instability. For the special case of a two-bus system we analytically derive a global stability map. 
\end{abstract}

\maketitle


\begin{quotation}
A reliable supply of electric power requires a stable operation of the electric power grid. Thousands of generators must run in a synchronous state with fixed voltage magnitudes and fixed relative phases. The ongoing transition to a renewable power system challenges the stability as line loads and temporal fluctuations increase. Maintaining a secure supply thus requires a detailed understanding of power system dynamics and stability. Among various models describing the dynamics of synchronous generators, analytic results are available mainly for the simplest second-order model which describes only the dynamics of nodal frequencies and voltage phase angles. In this article we analyze the stability of the third order model including the transient dynamics of voltage magnitudes. Within this model we provide analytical insights into the interplay of voltage and rotor-angle dynamics and characterize possible sources of instability. We provide novel stability criteria and support our studies with the analysis of a network of two coupled nodes, where a full analytic solution for the equilibria is obtained and a bifurcation analysis is performed.
\end{quotation}

\section{Introduction}

A stable supply of electric power is essential for the economy, industry and our daily life. The ongoing transition to a renewable generation challenges the stability of power grids in several ways \cite{Sims11}. Electric power has to be transmitted over large distances, leading to high transmission line loads at peak times \cite{Pesc14,Witthaut16}. Wind and solar power generation fluctuates on various time scales, requiring more flexibility and challenging dynamic stability \cite{Anvari16,Schmieten17,Schaefer17,Schafer18}. Furthermore, the effective inertia of the grid decreases such that power fluctuations have a larger impact on system stability \cite{Ulbig14}. 

A reliable power supply requires a detailed understanding of power system dynamics and stability. Numerical studies are carried out routinely at different levels of modelling detail (see Refs. \onlinecite{Mach08,Sauer98} for a comparison of different models). These studies provide a concrete stability assessment for one given power grid or components. For instance, the performance of different models has been evaluated for the Western System Coordinating Council System (WSCC) and the New England \& New York system
in Ref.~\onlinecite{Weck13}. Analytic studies into the mathematical structure of the problem have been obtained mainly for second-order models based on the power-swing equations \cite{Dorf10,Dorf13,Mott13,Schaefer15}. These models describe only the dynamics of nodal frequencies and rotor angles, assuming the voltage magnitudes to be constant in time \cite{Mach08,Berg81,Nish15}. Voltage stability is usually investigated numerically, see Ref. \onlinecite{Simpson16a} and references therein.

The present paper aims at providing analytical insights into the interplay of voltage and rotor-angle dynamics in electric power systems \cite{Kun04}. We analyze the stability of the third-order model of synchronous generators, including the transient voltage along the $q$-axis, which has been studied so far mainly computationally \cite{Mach08,Sauer98,Schm13,Jinp16,Ghah08,Dehg08, Karr04}. We provide an analytical decomposition of the Jacobian into the frequency and voltage subsystems, which gives rise to a novel stability criterion (cf. Proposition \ref{cor:schur}), and a characterization of possible sources of instability. For the most elementary network of two coupled nodes a full analytic solution of the equilibria is obtained and a bifurcation analysis is performed. 

The remainder of the paper is structured as follows: Section \ref{sec:problem} recalls the third-order model. In Section \ref{sec:local_stability} we perform a local stability analysis via the Jacobian linearization. 
Section \ref{sec:instability} discusses routes to instability, while in Section \ref{sec:example} draws upon the example of a two bus system. The paper ends with conclusions in Section \ref{sec:conclusions}.

\section{The third-order model and its equilibria} \label{sec:problem}

The third-order or one-axis model describes the transient dynamics of synchronous machines \cite{Mach08, Sauer98}, in particular the 
\begin{itemize}
\item
 power angle $\delta(t)$ relative to the grid reference frame,
\item
 the angular frequency $\omega(t) = \dot \delta$ relative to the grid reference frame and
\item
  the transient voltage $E'_q(t)$ in the $q$-direction of a co-rotating frame of reference.
\end{itemize} 
The third-order model does not cover subtransient effects and it assumes that the transient voltage in the $d$-direction of the co-rotating frame vanishes. The equations of motion for one machine are given by \cite{Mach08}
\begin{subequations} \label{eq:3rdorder-pre}
\begin{align}
   \dot \delta &= \omega, \\
   M \dot \omega &=   - D \omega + P^{m} - P^{el}, \\
   T'_d \dot E'_q &= E^{f} - E'_q + (X_d - X_d') I_d,
  \end{align}
\end{subequations}
where the dot denotes differentiation with respect to time. Here, the symbol $P^{m}$ denotes the effective mechanical input power of the machine and $E^{f}$ the internal voltage or field flux.  $P^{el}$ stands for electrical power out-flow. The parameters $D$ and $M$ denote the damping and the inertia of the mechanical motion and $T'_d$ the relaxation time of the transient voltage dynamics. The voltage dynamics further depends on the difference of the static ($X_d$) and transient ($X'_d$) reactances along the $d$-axis, where $X_d - X_d' > 0$ in general, and the current along $d$-axis $I_d$. 

In this article we consider an extended grid consisting of several synchronous machines labeled by $j=1,2,\ldots,N$. Neglecting transmission line losses, the electric power exchanged with the grid and the current at the $j$th machine read \cite{Schm13}
\begin{equation}
\begin{aligned}
    P^{el}_j &= \sum_{\ell=1}^N E'_{q,j} E'_{q,\ell} B_{j,\ell} \sin(\delta_j- \delta_\ell ), \\
    I_{d,j} &= \sum_{\ell=1}^N E'_{q,\ell} B_{j,\ell} \cos(\delta_\ell- \delta_j ),
    \label{eq:PI-network}
\end{aligned}
\end{equation} 
where the $E'_{q,j}$ and $\delta_j$ are the transient voltage and the power angle of the $j$th machine, the parameter $B_{jk} \ge 0$ denotes the susceptance of the transmission line $(j,k)$ and $B_{jj} \le 0$ denotes the shunt susceptance of the $j$th node. 

Using equations (\ref{eq:PI-network}), the equations of motion assume a particularly simple form \cite{Schm13,Jinp16,Schmieten17}. For the sake of notational convenience we drop the prime as well the subscripts $d$ and $q$ in the following and obtain
\begin{align}
   \dot \delta_j &= \omega_j,  \nn\\
 M_j \dot \omega_j &=  P^{m}_j - D_j \omega_j + \sum_{\ell=1}^N 
                       E_j E_\ell B_{j,\ell} \sin(\delta_\ell- \delta_j ),   \label{eq:3rdorder} \\
   T_j \dot E_j &= E_j^{f} - E_j + (X_j-X'_j) \sum_{\ell = 1}^N E_\ell B_{j,\ell} \cos(\delta_\ell-\delta_j).\nn
\end{align}

Many studies consider a grid consisting of synchronous generators and ohmic loads, which can then be eliminated using a Kron reduction \cite{Dorf11b}. The reduced system consists of the generator nodes only and the parameters $B_{jk}$ and $P^{m}_j$ represent effective values characterizing the reduced network.

The negligence of line losses is a common simplification in power grid stability assessment, as the ohmic resistance is typically much smaller than the susceptance in high-voltage power transmission grids. This assumption is not valid in distribution grids where resistance and susceptance are comparable. Furthermore, losses are expected to become more important when the transmitted power is large, i.e. at the border of the stability region.

Stationary operation of a power grid corresponds to a state with constant voltages and perfect phase-synchronization, i.e.~a point in configuration space where all $E_j$, $\omega_j$ and $\delta_j - \delta_\ell$ are constant in time. The latter condition requires that all nodes rotate at the same frequency $\delta_j(t) = \Omega t + \delta_j^\circ$ for all $j=1,\ldots,N$, such that we obtain the conditions
\be 
   \dot \omega_j = \dot E_j = 0, \quad \dot \delta_j =  \Omega \quad \mbox{for all } \; j=1,\ldots,N.
\ee
Strictly speaking, we are searching for a stable limit cycle, but all points along the cycle are physically equivalent. We can focus on any point on the cycle as a representative for the equivalence class and call this an equilibrium in the following. Subsequently, the superscript $\cdot^\circ$ is used to denote the values of the rotor phase angle, frequency and voltage in this equilibrium state.
Perturbations along the cycle, where we add or subtract a global phase shift $\delta$ from all phases $\delta_j$ simultaneously, do not affect phase synchronization and thus can be excluded from the stability analysis. We will make this precise in Definition \ref{def-stab}.

For the third-order model (\ref{eq:3rdorder}) an equilibrium state of the power grid is thus given by the nonlinear algebraic equations
\begin{subequations}   \label{eq:3rd-fixed}
\begin{align}
   \Omega &= \omega_j^{\circ},  \\
   0 &=  \underbrace{P^{m}_j -D_j \Omega}_{=: P_j}   + \sum_{\ell=1}^N 
                       E_j^{\circ} E_\ell^{\circ} B_{j,\ell} \sin(\delta_\ell ^{\circ}- \delta_j^{\circ} ),  \label{eq:3rd-fixed_b}\\
   0 &= E_j^{f} - E_j^{\circ} + (X_j-X'_j) \sum_{\ell = 1}^N E_\ell^{\circ} B_{j,\ell} \cos(\delta_\ell^{\circ}-\delta_j^{\circ}).  \label{eq:3rd-fixed_c}
\end{align}    
\end{subequations}
We note that many equilibria, stable and unstable, can exist in networks with sufficiently complex topology \cite{Manik17}.

\section{Linear Stability analysis} \label{sec:local_stability}

It is well-understood in the literature  \cite{Stro01, Kun04} that local stability properties of an equilibrium, i.e. stability with respect to small perturbations, can be evaluated by linearizing the equations of motion \eqref{eq:3rdorder}. For linear stability analysis we introduce perturbations $\xi$, $\nu$ and $\epsilon$: 
\be
   \delta_j(t) = \delta_j^{\circ} + \xi_j(t),  \; 
   \omega_j(t) =  \omega_j^{\circ} + \nu_j(t), \; E_j(t) = E_j^{\circ} + \epsilon_j(t). \nn 
\ee
The main question is then whether the perturbations $\xi$, $\nu$ and $\epsilon$ grow or decay over time. If all perturbations decay (exponentially), the equilibrium is said to be `linearly' stable. In the literature, this property is sometimes also called `local (asymptotic/exponential) stability' or `small system stability', cf. Refs. \onlinecite{Stro01, Khalil02}. We also refer to Ref. \onlinecite{Kun04} for an overview of stability notions for power systems. 
Substituting the ansatz from above into \eqref{eq:3rdorder}, transferring to a frame of reference rotating with the frequency $\Omega$ and keeping only terms linear in $\xi_j$, $\nu_j$ and $\epsilon_j$ yields
\begin{align}
    \dot \xi_j &= \nu_j, \nn  \\
    M_j \dot \nu_j &= -D_j \nu_j - \sum_{\ell=1}^N \Lambda_{j,\ell} \xi_\ell
                                            +                     \sum_{\ell=1}^N A_{\ell,j} \epsilon_\ell,  \label{eq:linstab1}   \\ 
    T_j  \dot \epsilon_j   &= - \epsilon_j + (X_j-X'_j) \sum_{\ell=1}^N H_{j,\ell} \epsilon_\ell 
           + \sum_{\ell=1}^N A_{j,\ell} \xi_\ell,    
          \nn                         
\end{align}
whereby 
$\vec \Lambda, \vec A, \vec H \in \mathbb{R}^{N \times N}$
are given by
\begin{align}
    \Lambda_{j,\ell} &= 
    \left\{ \begin{array}{l l}
        - E_j^{\circ} E_\ell^{\circ} B_{j,\ell} \cos(\delta_\ell^{\circ} - \delta_j^{\circ}) \; & \mbox{for} \, j \neq \ell \\
        \sum_{k\neq j}  E_j^{\circ} E_k^{\circ} B_{j,k} \cos(\delta_k^{\circ} - \delta_j^{\circ}) \; & \mbox{for} \, j = \ell \\
     \end{array} \right.  \nn \\
     A_{j,\ell} &= 
    \left\{ \begin{array}{l l}
        - E_\ell^{\circ} B_{j,\ell} \sin(\delta_\ell^{\circ} - \delta_j^{\circ}) \; & \hspace*{8mm} \mbox{for} \, j \neq \ell \\
        \sum_k E_k^{\circ} B_{j,k} \sin(\delta_k^{\circ} - \delta_j^{\circ}) \; & \hspace*{8mm} \mbox{for} \, j = \ell \\
     \end{array} \right.    \nn \\
    H_{j,\ell} &=  B_{j,\ell} \cos(\delta_\ell^{\circ} - \delta_j^{\circ}) .
     \label{eq:def_matrices}
\end{align}
Furthermore, we define the diagonal matrices $\vec M$, $\vec D$, $\vec X$ and $\vec T$ (all in $\mathbb{R}^{N \times N}$)  with elements $M_j$, $D_j$, $(X_j-X'_j)$ and $T_j$ for $j=1,\ldots,N$, respectively. We note that all these elements are strictly positive.

Now we can recast \eqref{eq:linstab1} into matrix form, defining the vectors $\vec \xi = (\xi_1,\ldots,\xi_N)^\top$, 
$\vec \nu = (\nu_1,\ldots,\nu_N)^\top$ and $\vec \epsilon = (\epsilon_1,\ldots,\epsilon_N)^\top$,
where the superscript $\top$ denotes the transpose of a matrix or vector.
We then obtain
\be
    \begin{pmatrix} \dot{\vec \xi} \\  \dot{\vec \nu} \\ \dot{\vec \epsilon} \end{pmatrix} \! = \!
    \underbrace{ \begin{pmatrix} \vec 0 & \eye & \vec 0 \\
                          -  \vec M^{-1} \vec \Lambda & - \vec M ^{-1} \vec D & \vec{M}^{-1} \vec A^\top \\
                           \vec T^{-1} \vec  X \vec A & \vec 0 &  \vec{T}^{-1} ( \vec X \vec H- \eye) \\
                 \end{pmatrix} }_{\displaystyle =: \vec J}        
    \!\! \begin{pmatrix} \vec \xi \\ \vec \nu \\ \vec \epsilon \end{pmatrix}.       
    \label{eq:eom-lin}                      
\ee
An equilibrium is dynamically stable if small perturbations are exponentially damped. This is the case if the real part of all relevant eigenvalues of the Jacobian matrix $\vec J$ are strictly smaller than zero \cite{Stro01}. If an eigenvalue $\mu$ has a positive real part then the corresponding eigenmode grows exponentially as $e^{\Re(\mu) t}$ and the system is linearly unstable. If the real part equals zero then the question for local stability cannot be decided using the linearization approach and a full nonlinear treatment based on a center-manifold approximation is necessary, see Ref. \onlinecite{Stro01}.
Typically, power systems are nonlinearly unstable in this case \cite{14bifurcation}.

We note that $\vec J$ always has one eigenvalue $\mu_1 = 0$ with the eigenvector 
$
    \left(\vec \xi ~ \vec \nu ~ \vec \epsilon \right)^\top
     = \left( \vec 1 ~ \vec 0 ~ \vec 0 \right)^\top.
     $
This corresponds to a global shift of the machines' phase angles which has no physical significance. We thus exclude this trivial eigenmode from the stability analysis in the following, i.e. we consider only perturbations from the orthogonal complement 
\be
    \mathcal{D}_\perp = \left\{ (\vec \xi, \vec \nu, \vec \epsilon)^\top \in \mathbb{R}^{3N} | 
      \vec 1 ^\top \vec \xi = 0 \right\}.
\ee
Similarly, we define the projection onto the angle and voltage subspaces (omitting frequency) 
\be
\tilde{\mathcal{D}}_\perp = \left\{ (\vec \xi, \vec \epsilon)^\top \in \mathbb{R}^{2N} | \vec 1 ^\top \vec \xi = 0 \right\},
\ee
and the angle subspace (omitting frequency and voltage) 
\be
\bar{\mathcal{D}}_\perp = \left\{ \vec \xi^\top \in \mathbb{R}^{N} | \vec 1 ^\top \vec \xi = 0 \right\}.
\ee
Furthermore, we fix the ordering of all eigenvalues such that
\be
  \mu_1 = 0, \qquad     \Re(\mu_2) \le \Re(\mu_3) \le \cdots \le \Re(\mu_{3N}).\label{orderEig}
\ee
We then have the following consistent definition of linear stability in all directions transversal to the limit 
cycle (cf. Ref. \onlinecite{Stro01}). 
\begin{defn}
The equilibrium $(\delta^\circ_j, \omega^\circ_j, E^\circ_j)$ is linearly stable if $\Re(\mu_n) < 0$  for all eigenvalues $n = 2,\ldots,3N$ of the Jacobian matrix $\vec J$ defined in (\ref{eq:eom-lin}).
\label{def-stab}
\end{defn}

We note that an eigenvalue $\mu_n$ can be complex, but then also the complex conjugate $\mu_n^*$ is an eigenvalue. Complex eigenvalues correspond to oscillatory modes, i.e. the variables oscillate around the equilibrium with a growing or shrinking amplitude. Indeed, the following lemma shows that only damped oscillations are possible. 
\begin{lemma}
\label{th:complex}
If an eigenvalue is complex $\Im(\mu_n) \neq 0$, then its real part is strictly negative $\Re(\mu_n) < 0$.
\end{lemma}

\begin{proof}
Recall the eigenvalue problem for the Jacobian
\[
    \vec J  \begin{pmatrix} \vec \xi \\ \vec \nu \\ \vec \epsilon \end{pmatrix}  
        = \mu  \begin{pmatrix} \vec \xi \\ \vec \nu \\ \vec \epsilon \end{pmatrix}.
    \label{eq:def-eigenvalue}    
\]
Decomposition yields
\begin{subequations} \label{eq:eigenvalue-comp}
\begin{align}
    \vec \nu &= \mu \vec \xi,  \label{eq:eigenvalue-comp_a}\\
     - \vec \Lambda \vec \xi - \vec D \vec \nu + \vec A^\top \vec \epsilon &= \mu \vec M \vec \nu,  \label{eq:eigenvalue-comp_b}\\
     \vec X \vec A \vec \xi  +  ( \vec X \vec H - \eye) \vec \epsilon &= \mu \vec T \vec \epsilon.  \label{eq:eigenvalue-comp_c}
\end{align}
\end{subequations}
Substituting \eqref{eq:eigenvalue-comp_a} into \eqref{eq:eigenvalue-comp_b} and multiplying  \eqref{eq:eigenvalue-comp_c} with $\vec X^{-1}$
yields
\begin{subequations}    \label{eq:tomultiply}
\begin{align}
   - \vec \Lambda \vec \xi - \mu \vec D \vec \xi + \vec A^\top \vec \epsilon 
        &= \mu^2 \vec M \vec \xi, \\
     \vec A \vec \xi + (\vec H - \vec X^{-1}) \vec \epsilon &= \mu \Vec X^{-1} \vec T \vec \epsilon .
\end{align}
\end{subequations}
Multiplying the two last equations with the Hermitian conjugates  $\vec \xi^\dagger$ and $\vec \epsilon^\dagger$ from the left, respectively, and equating $(\vec \xi^\dagger \vec A^\top \vec \epsilon)^\dagger = \vec \epsilon^\dagger \vec A \vec \xi$, one obtains
\begin{multline}
     \mu^2 \langle \vec \xi, \vec M \vec \xi \rangle  
      + \mu \langle \vec \xi, \vec D \vec \xi \rangle
      - \mu^* \langle \vec \epsilon, \vec T \vec X^{-1} \vec \epsilon \rangle  \\
     \qquad   + \langle \vec \xi,  \vec \Lambda  \vec \xi \rangle
            + \langle \vec \epsilon , ( \vec H - \vec X^{-1}  ) \vec \epsilon \rangle  = 0,
      \label{eq:quadratic2}      
\end{multline}
where $\langle \cdot,\cdot\rangle$ is the standard scalar product on $\mathbb{C}^N$. 
All matrices in these expressions are Hermitian such that all scalar products are real. Thus we can easily divide \eqref{eq:quadratic2}  into real and imaginary part obtaining
\be
   2 \Re(\mu) \Im(\mu)  \langle \vec \xi, \vec M \vec \xi \rangle 
    +  \Im(\mu) \big[ \langle \vec \xi, \vec D \vec \xi \rangle
    +  \langle \vec \epsilon, \vec T \vec X^{-1} \vec \epsilon \rangle \big] = 0. 
    \label{eq:mu-realimag}
\ee
This condition can be satisfied in two ways:
\begin{itemize}
\item[(1)] $\Im(\mu) = 0$. In this case we have an non-oscillatory mode and $\vec \xi$ and $\vec \epsilon$ can be chosen real. This case is analyzed in Lemma \ref{th:stability}.

\item[(2)] If $\Im(\mu) \neq 0$ we have a pair of complex conjugate eigenvalues and the corresponding eigenmode describes a damped or amplified oscillation. In this case we can divide equation (\ref{eq:mu-realimag}) by $\Im(\mu)$ and solve for $\Re(\mu)$ with the result
\be
   \Re(\mu) = - \frac{\langle \vec \epsilon, \vec T \vec X^{-1} \vec \epsilon \rangle  
      + \langle \vec \xi, \vec D \vec \xi \rangle}{
                           2 \langle \vec \xi, \vec M \vec \xi \rangle}.
\ee 
\end{itemize}
As all matrices $\vec T$, $\vec X$, $\vec D$ and $\vec M$ are diagonal with only positive entries we obtain that
$
    \Im(\mu) \neq 0   ~ \Rightarrow ~ \Re(\mu) < 0.
$
\end{proof}

As all oscillatory modes are damped, we can concentrate on real eigenvalues, $\mu\in\mathbb{R}$,  in the following.

\begin{lemma}
\label{th:stability}
The equilibrium $(\delta^\circ_j, \omega^\circ_j, E^\circ_j)$ is linearly stable if and only if the matrix
\be
   \vec \Xi = \begin{pmatrix}
      - \vec \Lambda & \vec A^\top \\
      \vec A & \vec H - \vec X^{-1}
      \end{pmatrix}
\ee  
is negative definite on $\tilde{\mathcal{D}}_\perp$. 
\end{lemma}

\begin{proof} By contradiction.
We first consider the case that $\vec \Xi$ is negative definite and show that this implies $\mu < 0$.
We start from \eqref{eq:tomultiply} for the eigenvalues $\mu$ and eigenvectors of $\vec J$, which is rearranged into matrix form:
\be
    (\vec \Xi + \vec C(\mu)) 
    \begin{pmatrix} \vec \xi \\ \vec \epsilon \end{pmatrix}
    = \vec 0,
    \label{eq:kernel-AM}
\ee 
where 
\be
     \vec C(\mu) =  \begin{pmatrix} 
     - \mu^2 \vec M - \mu \vec D  & 0  \\
     0  & - \mu \vec X^{-1} \vec T  \\
     \end{pmatrix} .
\ee
Now let $\vec \Xi$ be negative definite on $\tilde{\mathcal{D}}_\perp$. If we assume that $\mu \ge 0$ is an eigenvalue, the matrix $\vec C(\mu)$ is negative semi-definite. Thus the sum of the two matrices $\vec \Xi+ \vec C(\mu)$ is also negative definite such that equation (\ref{eq:kernel-AM}) cannot be satisfied. Thus an eigenvalue with $\mu \ge 0$ does not exist and the system is linearly stable.

Second, consider the case that the matrix $\vec \Xi$ is not negative definite on $\tilde{\mathcal{D}}_\perp$. We then find an eigenvalue $\beta \ge 0$ with corresponding eigenvector $(\vec \xi,\vec \epsilon)^\top$. We can assume that $\vec \epsilon \neq \vec 0$ because otherwise 

\be
\vec  \Xi \begin{pmatrix}
\vec \xi \\ \vec \epsilon
\end{pmatrix} = \beta \begin{pmatrix}
\vec \xi \\ \vec \epsilon
\end{pmatrix}
\ee
would imply that $\vec \Lambda \vec \xi = \vec 0 $ and it would mean that either $\vec \xi = \vec 0$ or $\vec \xi \propto \vec 1$. The first option is impossible because this would imply $(\vec \xi, \vec \epsilon)^\top= \vec 0 $, the second is ruled out by the assumption that $(\vec \xi, \vec \epsilon)^\top \in \tilde{\mathcal{D}}_\perp 
$.

Then we evaluate the expression
\begin{align}
    & \begin{pmatrix} \vec \xi \\ \vec 0 \\ \vec \epsilon \end{pmatrix}^\top \vec J
    \begin{pmatrix} \vec \xi \\ \vec 0 \\ \vec \epsilon \end{pmatrix} \nn \\
    & = \begin{pmatrix} \vec \xi \\ \vec 0 \\ \vec \epsilon \end{pmatrix}^\top \begin{pmatrix} \vec \eye & 0 & 0 \\ 0 & \vec M^{-1}  & 0 \\   0 & 0 & \vec T^{-1} \vec X \end{pmatrix} \begin{pmatrix}
      0 & \vec \eye & 0 \\
      -\vec \Lambda & -\vec D & \vec A^\top \\
      \vec A & 0 & \vec H - \vec X^{-1}
      \end{pmatrix}
          \begin{pmatrix} \vec \xi \\ \vec 0 \\ \vec \epsilon \end{pmatrix} \nn \\
              &=  \begin{pmatrix} \vec \xi \\ \vec 0 \\ \vec \epsilon \end{pmatrix}^\top \begin{pmatrix} \vec \eye & 0 & 0 \\ 0 & \vec M^{-1}  & 0 \\   0 & 0 & \vec T^{-1} \vec X \end{pmatrix} \beta
          \begin{pmatrix} \vec 0 \\ \vec \xi \\ \vec \epsilon \end{pmatrix} \nn \\ 
    & =  \begin{pmatrix} \vec \xi \\ \vec 0 \\ \vec \epsilon \end{pmatrix}^\top \begin{pmatrix}
    0 \\ \beta \vec M^{-1} \vec \xi \\ \beta \vec T^{-1} \vec X \vec \epsilon \end{pmatrix}  = \beta \vec \epsilon^{\top} \vec T^{-1} \vec X \vec \epsilon \ge  0,
\end{align}
whereby the last inequality follows from $\vec X, \vec T$ being positive definite.
Thus the Jacobian is not negative definite on $\mathcal{D}_\perp$ and the equilibrium is not linearly stable.
\end{proof}

\section{Routes to instability} \label{sec:instability}

We can obtain further insights into the stability properties by decomposing the system dynamics into its rotor-angle and voltage parts (cf. Refs. \onlinecite{Kun04, Vou96}). We start with the following proposition, which comes in two versions:
\begin{prop}[Sufficient and necessary stability conditions]
\label{cor:schur} ~\vspace*{-4mm}\\
\begin{itemize}
\item[I.]The equilibrium $(\delta^\circ_j, \omega^\circ_j, E^\circ_j)$ is linearly stable if and only if 
(a) the matrix $\vec \Lambda$ is positive definite on $\bar{\mathcal{D}}_\perp$ and
(b) the matrix $\vec H - \vec X^{-1} + \vec A \vec \Lambda^{+} \vec A^\top$ is negative definite,
where $\vec \Lambda^{+}$ is the Moore-Penrose pseudoinverse of $\vec \Lambda$.

\item[II.] The equilibrium $(\delta^\circ_j, \omega^\circ_j, E^\circ_j)$ is linearly stable if and only if
(a) the matrix $\vec H - \vec X^{-1}$ is negative definite and
(b) the matrix $\vec \Lambda + \vec A^\top (\vec H - \vec X^{-1})^{-1} \vec A$ is positive definite on $\bar{\mathcal{D}}_\perp$.
\end{itemize}
\end{prop} 

\begin{proof}
The results follow from Lemma \ref{th:stability} using the Schur complement \cite{Zhan06}. We demonstrate the proof for Part I. The proof for criterion II is obtained analogously.

First, on $\tilde{\mathcal{D}}_\perp$, the matrix $\vec \Xi$ can be decomposed  into
\begin{align}
      & \vec \Xi = \vec{U}^\top 
          \begin{pmatrix}
      - \vec \Lambda & 0 \\
      0 & \vec H - \vec X^{-1} + \vec A \vec \Lambda^{+} \vec{A}^\top
      \end{pmatrix} \vec U ,
      \label{eq:Xi-Schur}
\end{align}
with
\begin{align}
          \vec U = \begin{pmatrix} \eye & -\vec \Lambda^{+} \vec{A}^\top  \\ 0 & \eye \end{pmatrix}. 
          \label{eq:defU}
\end{align}
We define
\begin{align}
    \vec y = \begin{pmatrix} \vec y_a & \vec y_v \end{pmatrix}^\top
     = \vec U \vec x.
    \label{eq:yPx}
\end{align}
Lemma \ref{lem_for_theorem} given in the Appendix shows that
\begin{align}
    \vec x \in \tilde{\mathcal{D}}_\perp \quad &\Leftrightarrow \quad \vec y \in \tilde{\mathcal{D}}_\perp, \nn \\
    \vec x \neq \vec 0 \quad &\Leftrightarrow \quad \vec y_a, \vec y_v \neq 0.
    \label{eq:xy-definite}
\end{align}
According to Lemma \ref{th:stability}, an equilibrium is stable if and only if
\begin{align*}
   \vec{x}^\top \vec \Xi \vec x < 0 \qquad \forall \vec x \in \tilde{\mathcal{D}}_\perp, \quad \vec x \neq 0.
\end{align*}
By the use of equations \eqref{eq:Xi-Schur} and \eqref{eq:xy-definite} the above condition is equivalent to
\begin{align*}
   - \vec y_a^\top \vec \Lambda \vec y_a  
   + \vec y_v^\top(\vec H - \vec X^{-1} + \vec A \vec \Lambda^{+} \vec{A}^\top) \vec y_v < 0 ,\nn \\
   \quad \forall \, \vec y \in \tilde{\mathcal{D}}_\perp, \quad \vec y_a,\vec y_v \neq 0.
\end{align*}
That is, the equilibrium  $(\delta^\circ_j, \omega^\circ_j, E^\circ_j)$ is stable if and only if $\vec \Lambda$ is positive definite on $\bar{\mathcal{D}}_\perp$ and
$\vec H - \vec X^{-1} + \vec A \vec \Lambda^{+} \vec{A}^\top$ is negative definite.
\end{proof}

We are especially interested in classifying possible routes to instability in the third-order model, i.e. we want to understand how a stable equilibrium can be lost if the system parameters are varied. Proposition \ref{cor:schur}  provides the basis for this task as it shows where instabilities emerge. Depending on which of the criteria in the corollary is violated, the instability can be attributed to the angular or the voltage subsystem or both. This analysis will furthermore reveal stabilizing factors of the third-order model. 

\subsection{Rotor-angle stability}
\label{sec:angle-stab}

An equilibrium becomes unstable when the Condition I(a) in Proposition \ref{cor:schur} is violated, i.e. when the matrix $\vec \Lambda$ is no longer positive definite on $\bar{\mathcal{D}}_\perp$. This corresponds to a pure angular instability which can be seen as follows. If we artificially fix the voltages of all nodes we arrive at a second-order model, commonly referred to as classical model \cite{Mach08}, structure-preserving model \cite{Berg81} or oscillator model \cite{Fila08,12powergrid} in the literature. In this case we have $\vec \epsilon \equiv 0$ and the linearized equations of motion (\ref{eq:eom-lin}) reduce to      
\[
    \begin{pmatrix} \dot{\vec \xi} \\ \dot{\vec \nu} \\  \end{pmatrix} \! = \!
     \begin{pmatrix} \vec 0 & \eye &  \\
                          -  \vec M^{-1} \vec \Lambda & - \vec M ^{-1} \vec D  \\
                 \end{pmatrix}    
    \!\! \begin{pmatrix} \vec \xi \\ \vec \nu \\  \end{pmatrix}                           
\]
This system is  stable if and only if $\vec \Lambda$ is positive definite (cf. Ref. \onlinecite{14bifurcation}). Thus 
Condition I(a) includes a condition for the stability of the rotor-angle subsystem alone.

For a system of two machines $j=1,2$ connected by one transmission line this stability criterion is easily understood. The matrix $\vec \Lambda$ is positive definite on $\bar{\mathcal{D}}_\perp$ if and only if the phase difference satisfies $\cos(\delta_2^{\circ} - \delta_1^{\circ})>0$. This is a well-known criterion for stability extensively discussed in the literature \cite{Mach08, Sauer98, Chen20161}. Stability is lost if the real power transmitted between the two nodes,
\[
     P_{1} =  E_1^{\circ} E_2^{\circ} B_{1,2} \sin(\delta_1 ^{\circ}- \delta_2^{\circ} ),
\]
is strongly increased. Typically, the phase difference $\delta_1 ^{\circ}- \delta_2^{\circ}$ grows with $P_{1}$ until it reaches $\pi/2$ where the grid becomes unstable. 

An example of such a pure angle instability is shown in Fig.~\ref{fig:ins-freq}. For $X_j - X_j' = 0$ the voltages are fixed as $E_j^\circ = E_j^{f}$ such that we can restrict our analysis to the stability of the rotor-angle subsystem. Stable operation is possible as long as $P_{1} < P_{\rm max} = B_{1,2} E_1^{f} E_2^{f}$ as elaborated above. As $P_{1} \rightarrow P_{\rm max}$, the second eigenvalue of the Laplacian $\vec \Lambda$ approaches zero and the stable equilibrium is lost in a saddle-node bifurcation as studied in detail in Refs. \onlinecite{14bifurcation,Jinp16}.

 \begin{figure}[tb]
\centering
\includegraphics[width=\columnwidth,trim= 1cm 0.5cm 2cm 1.5cm, clip]{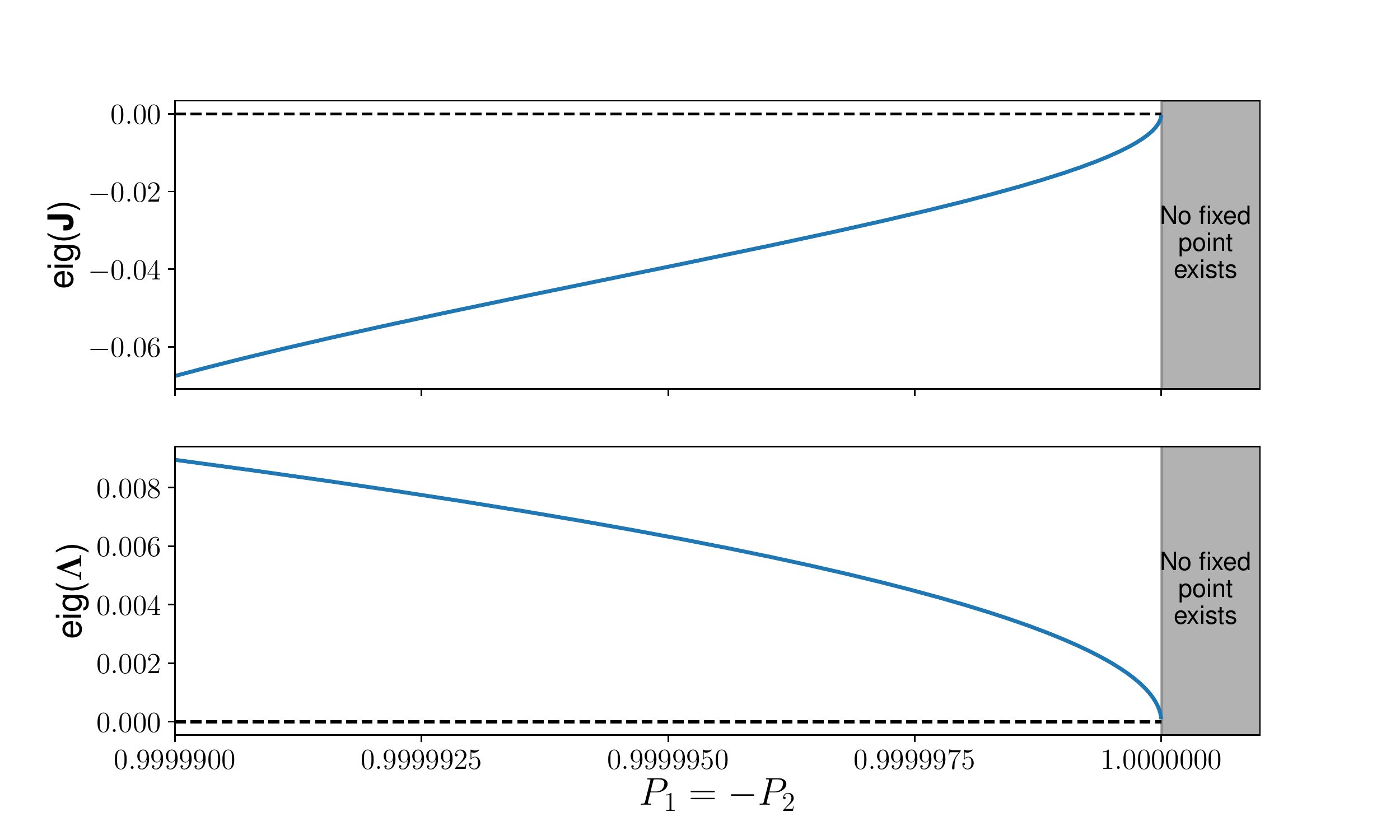}
\caption{
\label{fig:ins-freq}
Example of a pure angle instability in a system of two coupled machines. When the transmitted power $P_1 = - P_2$ is increased towards a critical value $P_{\rm max}$, one eigenvalue of the Laplacian matrix $\vec \Lambda$ (shown on the lower panel) approaches zero. At once one eigenvalue of the Jacobian $\vec J$ of the full dynamical system approaches zero and stability is lost as described by Proposition \ref{cor:schur}. At the critical value the equilibrium is lost in an inverse saddle-node bifurcation, corresponding to the border of the regions I and II in the stability map in Fig.~\ref{fig:stab-map}.
The remaining parameters are $X_1 - X'_1 = X_2 -X'_2 =0$, $B_{1,1} = B_{2,2} = -0.8$, $B_{1,2} = B_{2,1} = 1.0$, $E^{f}_1 = E_2^{f} = 1.0, D = 0.2, M = 1, T = 2$ in per unit system.
}
\end{figure}

This result can be generalized to larger networks in the following way. We assume that the network is connected, otherwise we can consider every connected component separately. A transmission line $(i,j)$ is called a `critical line' if the angle difference across $(i,j)$ satisfies
\begin{equation*}
\dfrac{\pi}{2} \le |\theta_i - \theta_j|\mod 2\pi \le \dfrac{3\pi}{2} ,
\end{equation*}
otherwise it is called a non-critical line. We then have the following corollary.

\begin{corr}
\label{cor:laplacian-pos}
If there is no critical line in the network, then the matrix $\vec \Lambda$ is positive definite on $\bar{\mathcal{D}}_\perp$.
\end{corr}

This result follows from the fact that $\vec \Lambda$ is a Laplacian matrix of a graph with weights $w_{ij} = E_i^{\circ} E_j^{\circ} B_{i,j} \cos(\delta_j^{\circ} - \delta_i^{\circ})$. If no critical line exists, $\vec \Lambda$
is a conventional Laplacian, which is extensively studied in the literature \cite{Merris94,Fied73, Godsil01}. In particular, all eigenvalues are positive except for one zero eigenvalue corresponding to the eigenvector $\vec 1$, i.e. a global shift of the voltage phase angles which has no physical significance as discussed above.

The general case of a signed Laplacian with possibly negative weights $w_{ij}$ has been studied in detail only recently \cite{14bifurcation,Chen20161,Chen20162,Song15}. Typically only few critical lines are present such that it is insightful to analyze how stability depends on the properties of these edges \cite{Song15}. One can draw an analogy to resistor networks, where $1/w_{ij}$ is taken as the direct resistance of each edge in the network. The effective resistance between a pair of nodes is then defined as the voltage drop inducing a unit current between the respective nodes. 

If only one critical line is present, its direct resistance is negative. However, the Laplacian matrix remains positive semidefinite as long as the effective resistance of the critical line remains positive \cite{Zelazo14}.

This result has been generalized to the case of several critical lines as follows. Let $G_-$ be the subgraph containing all vertices but only the critical edges and let $F_-$ be a spanning forest of this subgraph. We denote the node-edge incidence matrix of the spanning forest as $\vec  D_-$. Then the matrix
\[
   \vec \Gamma_- = \vec  D_- \vec \Lambda^+  \vec D_-^\top
\]
is an effective resistance matrix: the diagonal elements are the effective resistances of the respective edges and the non-diagonal elements are mutual effective resistances. One then finds:
\begin{corr}
\label{cor:laplacian-pos2}
The signed Laplacian $\vec \Lambda$ is positive definite on $\bar{\mathcal{D}}_\perp$ if and only if the matrix $\vec \Gamma_- $ is positive definite\cite{Chen20162}.
\end{corr}

\subsection{Voltage stability}

Condition II (a) in Proposition \ref{cor:schur} describes the stability of the voltage subsystem alone. To see this, we artificially fix the phase angles and frequencies, i.e. we set $\vec \xi \equiv \vec \nu \equiv 0$ and consider only perturbations of the voltages $\vec \epsilon$.
The linearized equations of motion (\ref{eq:eom-lin}) then reduce to 
\[
   \dot{ \vec \epsilon} =  \vec{T}^{-1} \vec X (\vec H- \vec X^{-1}) \vec \epsilon .
\]
This system is stable if and only if the matrix $\vec H - \vec X^{-1}$ is negative definite. If this condition is violated we thus face a pure voltage instability. Rotor angles and frequencies and phase angles will be affected by such an instability, too, but the origin lies in the voltage subsystem. 

To obtain a deeper understanding of this condition, we first consider a system of only two machines $j = 1, 2$ connected by one transmission line. Then we have
\[
    \vec H - \vec X^{-1} = 
    \begin{pmatrix}
    B_{1,1} - (X_{1} - X_{1}')^{-1} & B_{1,2}\cos(\delta_{1}^{\circ} - \delta_{2}^{\circ}) \\
    B_{2,1}\cos(\delta_{2}^{\circ} - \delta_{1}^{\circ}) & B_{2,2} - (X_{2} - X_{2}')^{-1} \\
    \end{pmatrix}.
\]
The cosine is positive since otherwise the rotor-angle instability would arise; $X_{j} - X_{j}'$ and $B_{j,\ell}$ are positive for $j\neq\ell$ while $B_{j,j} \le 0$. Applying Silvester's criterion we find that $\vec H - \vec X^{-1}$ is negative definite if and only if
\begin{align}
  &\left(B_{1,1} - (X_{1} - X_{1}')^{-1} \right)
  \left(B_{2,2} - (X_{2} - X_{2}')^{-1} \right) \nn \\
  & \qquad >  B_{1,2} B_{2,1} \cos^2(\delta_{1}^{\circ} - \delta_{2}^{\circ}).
  \label{eq:HXnegdef}
\end{align}
For applications it is desirable to find sufficient stability conditions which depend purely on the machine parameters and not on the state variables. Using the bound $\cos^2(\delta_{1}^{\circ} - \delta_{2}^{\circ}) \le1$, a sufficient condition for voltage stability is obtained:
\begin{align*}
  \left(B_{1,1} - (X_{1} - X_{1}')^{-1} \right)
  \left(B_{2,2} - (X_{2} - X_{2}')^{-1} \right) >  B_{1,2} B_{2,1}. \nn
\end{align*}
This condition is satisfied if 
\begin{align}
   (X_1 -X_1') &< (B_{1,1} + B_{1,2})^{-1} \qquad \mbox{and} \nn \\
   (X_2 -X_2') &< (B_{2,1} + B_{2,2})^{-1}.
   \label{eq:volt2-suff2}
\end{align}
Thus, voltage stability is threatened if the difference between the static and transient reactances  $(X_j-X_j')$ becomes too large. An example for this route to instability is shown in Fig.~\ref{fig:bifI-III}.

\begin{figure}[tb]
\centering
\includegraphics[width=1\columnwidth,trim=0.5cm 1.2cm 2cm 2cm,clip]{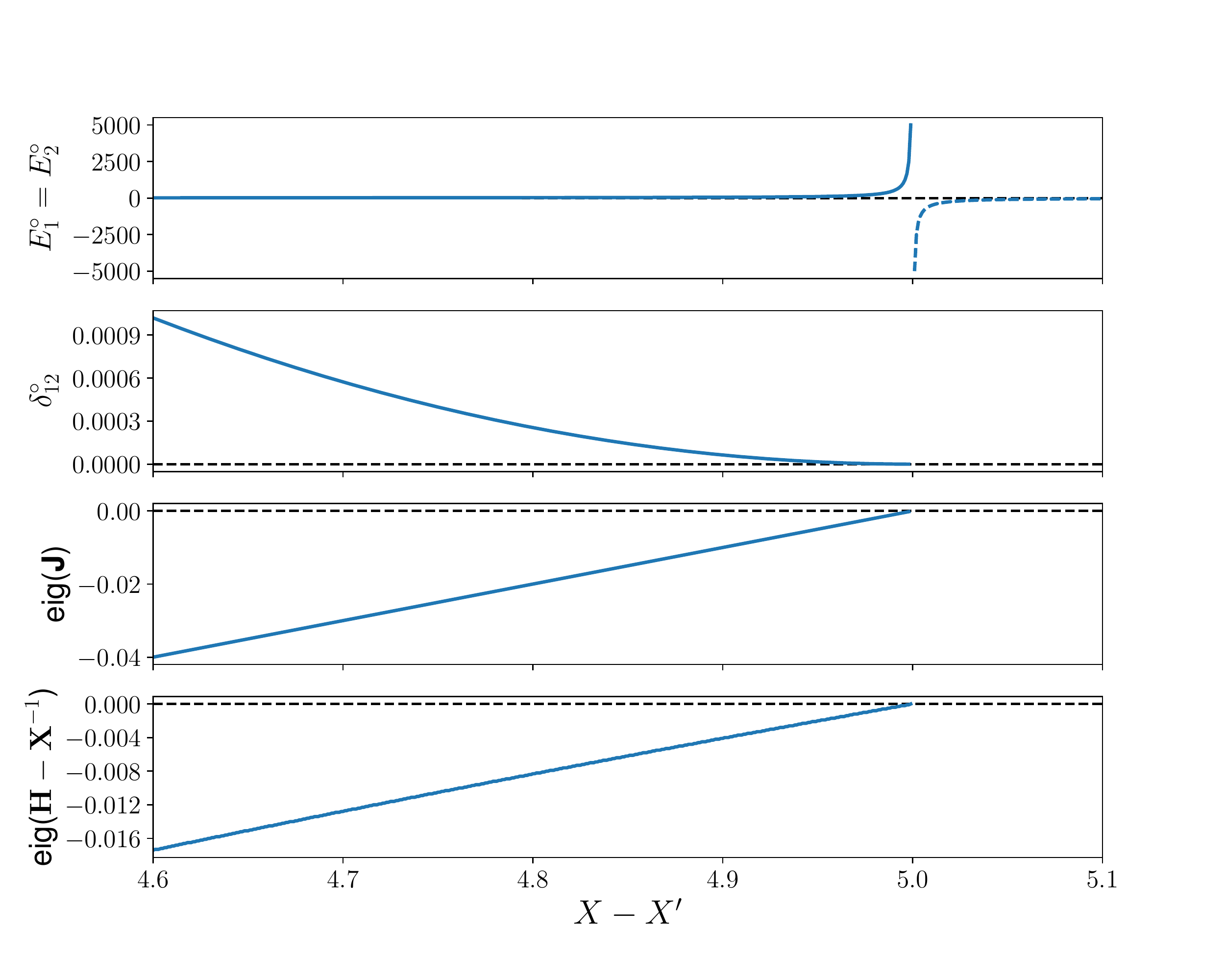}
\caption{
\label{fig:bifI-III}
Example of a voltage instability in a system of two coupled identical machines. As the parameter $X-X'$ approaches the critical value given by Eq.~(\ref{eq:HXneg-example}), the nodal voltages tend to infinity and the eigenvalues of matrices $\vec H - \vec X^{-1}$ and $\vec{J}$ tend to 0. The instability corresponds to the border of the regions I and III in the global stability map in Fig.~\ref{fig:stab-map}.
The remaining parameters are $P_1 =-P_2=0.5$, $B_{1,1} = B_{2,2} = -0.8$, $B_{1,2} = B_{2,1} = 1.0$, $E^{f}_1 = E_2^{f} = 1.0, D = 0.2, M = 1, T = 2$ in per unit system.
}
\end{figure}

In more detail, we consider two identical machines running idle, i.e. no power is transmitted $P_1 = P_2 = 0$. In this case the sufficient conditions 
(\ref{eq:volt2-suff2}) for stability are also necessary conditions and the stable equilibrium is lost at the critical value
\begin{align}
      (X-X')_{\rm crit} = (B_{1,1} + B_{1,2})^{-1}.
  \label{eq:HXneg-example}
\end{align}
When $X-X'$ is increased above this critical value, an eigenvalue of the matrix $\vec H - \vec X^{-1}$ crosses zero. Then also an eigenvalue of the Jacobian $\vec J$ of the full dynamical system crosses zero and the equilibrium becomes unstable, as described by Proposition \ref{cor:schur}. 

This result can be generalized to  larger systems comprised of more than two machines. Next, we present a sufficient and a necessary conditions for the stability of the voltage subsystem alone. Again, we will support our proposal for small values of $(X_{j} - X_{j}')$ for ensuring voltage stability.

\begin{corr}
\label{cor:voltage1}
If, for all nodes $j=1,\ldots,N$, 
\[ 
    (X_j - X_j')^{-1} > \sum_{\ell=1}^N B_{j,\ell},
\]
then the matrix $\vec H - \vec X^{-1}$ is negative definite.
\end{corr}
\begin{proof}
By applying Gershgorin's circle theorem \cite{Gers31} to the general form of the matrix $\vec H - \vec X^{-1}$, the following condition for eigenvalues $\lambda_j$ is obtained:
\[
|\lambda_{j} - (B_{j,j} - (X_j - X_j')^{-1})| \leq |\sum_{\ell \neq j}^N B_{j,\ell} \cos{(\delta_{j}^{\circ} - \delta_{l}^{\circ})}|,
\]
with the radius of the disk $\sum_{\ell \neq j}^N B_{j,\ell} \cos{(\delta_{j}^{\circ} - \delta_{l}^{\circ})}$ centered at $B_{j,j} - (X_j - X_j')^{-1}$. The matrix is negative definite if and only if all the eigenvalues lie in the left half of the complex plane, which is guaranteed if
\[
B_{j,j} - (X_j - X_j')^{-1}+\sum_{\ell \neq j}^N B_{j,\ell} \cos{(\delta_{j}^{\circ} - \delta_{l}^{\circ})} < 0.
\]
Using the bound $\cos{(\delta_{j}^{\circ} - \delta_{l}^{\circ})} \le 1$, a sufficient condition for negative definiteness is obtained as
\begin{multline*}
    B_{j,j} - (X_j - X_j')^{-1}+\sum_{\ell \neq j}^N B_{j,\ell} < 0 \\
    \Leftrightarrow  (X_j - X_j')^{-1} > \sum_{\ell =1}^N B_{j,\ell} \, .
\end{multline*}
This concludes the proof. \end{proof}

\begin{corr}
\label{cor:voltage2}
If for any subset of nodes $\mathcal{S} \subset \{1,2,\ldots,N\}$,
\be
    \sum_{j \in \mathcal{S}} (X_j - X_j')^{-1} \le \sum_{j,\ell \in \mathcal{S}} 
                     B_{j,\ell} \cos(\delta_\ell^\circ - \delta_j^\circ),
\ee
then the matrix $\vec H - \vec X^{-1}$ is not negative definite and the equilibrium is linearly unstable.
\end{corr}
\begin{proof}
This result follows from evaluating the expression $\vec x^\top (\vec H - \vec X^{-1}) \vec x$ for a trial vector $\vec x \in \mathbb{R}^N$ with entries $x_j = 1\; \forall j \in \mathcal{S}$ and $x_j = 0\; \forall j \notin \mathcal{S}$. 
\end{proof}

\subsection{Mixed instability}

\begin{figure}[tb]
\centering
\includegraphics[width=1\columnwidth, trim= 0.5cm 1.5cm 2cm 3cm, clip]{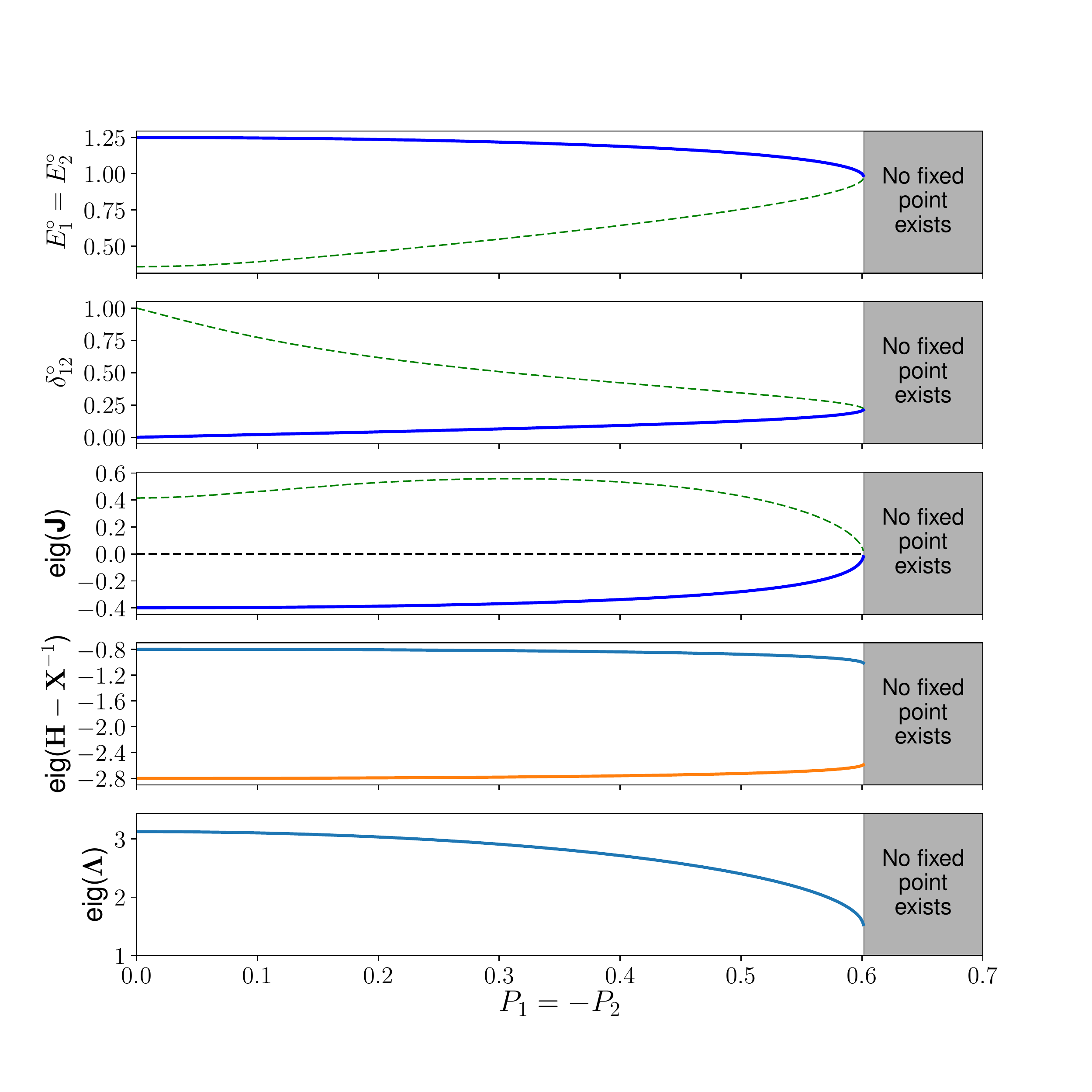}
\caption{
\label{fig:bifI-II}
Example of a mixed instability in a system of two coupled machines. As the transmitted power $P_{1} = - P_{2}$ increases to a critical value, a stable equilibrium (solid lines) and an unstable equilibrium (green dashed line) annihilate in an inverse saddle node bifurcation. The bifurcation corresponds to the border of the regions I and II in the stability map in Fig.~\ref{fig:stab-map}.
Shown are the state variables $E^\circ_{1} = E^\circ_2$ and the relative phase $\delta_{12}^\circ$ and the eigenvalues describing linear stability. Only the largest eigenvalues are shown for $\vec \Lambda$ and $\vec J$.  The remaining parameters are $X_1 - X'_1 = X_2 -X'_2 =1$, $B_{1,1} = B_{2,2} = -0.8$, $B_{1,2} = B_{2,1} = 1.0$, $E^{f}_1 = E_2^{f} = 1.0, D = 0.2, M = 1, T = 2$ in per unit system.
}
\end{figure}

The interplay of voltage and angle dynamics can lead to a third type of instability. A genuine mixed instability is observed if the Conditions I (a) and II (a) in Proposition \ref{cor:schur} are satisfied such that no `pure' angle or voltage instability occur, but one of the Conditions I (b) or II (b) is violated. An example of a mixed instability is shown in Fig.~\ref{fig:bifI-II}. The matrices quantifying voltage stability ($\vec H - \vec X^{-1}$) and angle stability ($\vec{\Lambda}$) remain negative and positive definite, respectively, but still stability is lost in a saddle node bifurcation when the transmitted power is increased.

It should be noted that a mixed instability is not exceptional, but the typical case. In the case of two coupled machines a pure angle instability is observed only if $X_j-X_j'=0$ for $j = 1, 2$ (cf. Fig.~\ref{fig:ins-freq}). If $X_j-X_j'>0$ an increase of the transmitted power can only lead to a mixed instability.

The emergence of mixed instabilities demand more rigid requirements for stability as discussed above. Based on proposition \ref{cor:schur}, we derive necessary and sufficient conditions in terms of network connectivity measures.

\begin{defn} 
Given the ordering \eqref{orderEig}, the second smallest eigenvalue $\lambda_2$ of the Laplacian matrix $\vec \Lambda$ is known as the Fiedler value or algebraic connectivity of a network\cite{Fied73}. The associated eigenvector $\vec{v}_F$ is the Fiedler vector\cite{Fied75}.
\label{Fiedler}
\end{defn}

The algebraic connectivity $\lambda_2$ is a measure of the connectivity of a graph, embodying its topological structure and connectedness. For a conventional graph with positive edge weights the algebraic connectivity is greater than $0$ if and only if the graph is connected (cf.~\ref{cor:laplacian-pos}).


In the limit $(X_j - X'_j) \equiv 0$ (no voltage dynamics) a necessary and sufficient condition for stability is that the Laplacian $\vec \Lambda$ is positive definite on $\bar{\mathcal{D}}_\perp$ which is equivalent to
\be 
    \lambda_2 \stackrel{!}{>} 0
\ee
as discussed in detail in section \ref{sec:angle-stab}. We can extend this necessary condition for the algebraic condition for small $(X_j - X'_j)$ by means of a Taylor expansion. To leading order, the necessary connectivity always increases with $(X_j - X'_j)$.

\begin{corr}
A necessary condition for the stability of an equilibrium point is given by
\be
   \lambda_2 > \sum_{j=1}^{N} (X_j - X'_j) \left(\sum_{k=1}^{N} A_{jk} v_{Fk}\right)^2
        + \mathcal{O}((X_j - X'_j)^2),\nn
\ee
where $\vec v_F$ denotes the Fiedler vector for $(X_j - X'_j) \equiv 0$.
\end{corr}

\begin{proof}
We denote by $\vec v_F$ the normalized Fiedler vector at $(X_j - X'_j) \equiv 0$ and by $\vec v'_F$ the actual normalized Fiedler vector for the particular non-zero value of the $(X_j - X'_j)$. Clearly, we have
\be 
   \vec v'_F = \vec v_F + \mathcal{O}((X_j - X'_j)^1).\nn 
\ee
Furthermore, we use the expansion
\be 
   (\vec X^{-1}- \vec H)^{-1} =  \sum_{\ell = 0}^\infty \vec X(\vec X \vec H)^\ell,\nn
\ee
such that 
\be 
   (\vec X^{-1}- \vec H)^{-1} = \vec X + \mathcal{O}((X_j - X'_j)^2). \nn
\ee
Now, condition II.(b) in proposition \ref{cor:schur} reads
\be 
    \forall \vec y: \quad 
    \vec y^\top \vec \Lambda \vec y > 
    \vec y^\top \vec A^\top (\vec X^{-1}- \vec H)^{-1} \vec A \vec y.\nn
\ee
Choosing one particular vector $\vec y$ one obtains a necessary condition for stability. Picking $\vec y = \vec v'_F$ we find
\be 
   \lambda_2 > \vec {v'}_F^{\top} \vec A^\top 
      (\vec X^{-1}- \vec H)^{-1} \vec A \vec v'_F ,\nn
\ee
and expanding the right-hand side to leading order in $(X_j - X'_j)$ yields
\be 
  \lambda_2 > \vec v_F^\top \vec A^\top \vec X \vec A \vec v_F 
     + \mathcal{O}((X_j - X'_j)^2).\nn
\ee
The matrix $\vec X$ is diagonal such that we can evaluate the expression further to obtain
\be 
   \lambda_2 > \sum_{j=1}^{N} (X_j - X'_j) \left(\sum_{k=1}^{N} A_{jk} v_{Fk} \right)^2 
     + \mathcal{O}((X_j - X'_j)^2).\nn
\ee 
\end{proof}


Furthermore, we can derive two sufficient conditions for stability in terms of the algebraic connectivity. These results show that in the limit of high connectivity only the pure voltage dynamics determines the stability of the full system.

\begin{corr}
If the algebraic connectivity is positive $\lambda_2>0$ and for all nodes $j=1,\ldots,N$ we have
\be   
    \label{eq:stab-cond-Xl}
   (X_j - X'_j)^{-1} - \sum_{\ell=1}^N B_{j \ell} > 
     \frac{\| \vec A \|_{2}  \| \vec A^\top \|_{2}}{\lambda_2},
\ee
where $\| \cdot \|_{2}$ is the operator $\ell_2$-norm, then an equilibrium point is stable.
\end{corr}
 
\begin{proof}
(1) If $\lambda_2>0$ this directly implies that $\vec \Lambda$ is positive definite on $\bar{\mathcal{D}}_\perp$ and condition I.(a) in proposition \ref{cor:schur} is satisfied.
 
(2) Using Gershgorin's circle theorem as in the proof of corollary \ref{cor:voltage1} we find that the condition (\ref{eq:stab-cond-Xl}) implies that
$
    \vec X^{-1} - \vec H 
      - \lambda_2^{-1}\| \vec A \|_{2}  \| \vec A^\top \|_{2} \eye 
$ is positive definite on $\bar{\mathcal{D}}_\perp$. Noting that $\lambda_2^{-1} = \| \vec \Lambda^+ \|_{2}$, this implies
\begin{align}
  \forall \vec y: \quad
        \vec y^\top (\vec X^{-1} - \vec H) \vec y 
        & > \| \vec A \|_{2} \| \vec \Lambda^+ \|_{2} \| \vec A^\top \|_{2}
             \| \vec y \|^2 \nn \\
        & \ge \| \vec A  \vec \Lambda^+ \vec A^\top \|_{2} \| \vec y \|^2 \nn \\
        & \ge \vec y^\top \vec A  \vec \Lambda^+ \vec A^\top \vec y.
\end{align}
Hence, matrix $\vec X^{-1} - \vec H - \vec A  \vec \Lambda^+ \vec A^\top$ is positive definite on $\bar{\mathcal{D}}_\perp$. Condition I.(b) in proposition \ref{cor:schur} is satisfied and the equilibrium is stable.
\end{proof}

\begin{corr}
If the voltage subsystem is stable, i.e.~$\vec X^{-1} - \vec H$ is positive definite, and the algebraic connectivity satisfies
\be   
    \label{eq:stab-cond-lam}
    \lambda_2 > \| \vec A^\top  \vec (\vec X^{-1} - \vec H)^{-1}  \vec A \|_2,
\ee
where $\| \cdot \|_2$ is the operator $\ell_2$-norm, then an equilibrium point is stable.
\end{corr}

\begin{proof} 
(1) By assumption, we have $\vec X^{-1} - \vec H$ positive definite such that condition II.(a) is satisfied in proposition \ref{cor:schur}.

(2) Furthermore, the assumption (\ref{eq:stab-cond-lam}) implies that
\begin{align}
   \forall \vec y \in \bar{\mathcal{D}}_\perp: \quad
   \vec y^\top \vec \Lambda \vec y & \ge \lambda_2 \|\vec y\|^2 \nn \\
   &> \| \vec A^\top  \vec (\vec X^{-1} - \vec H)^{-1}  \vec A \|_2 
          \|\vec y\|^2 \nn \\
   &\ge \vec y^\top \vec A^\top  \vec (\vec X^{-1} - \vec H)^{-1}  \vec A \vec y, \nn
\end{align}
such that the matrix $\vec \Lambda - \vec A^\top  \vec (\vec X^{-1} - \vec H)^{-1}  \vec A$ is positive definite on $\bar{\mathcal{D}}_\perp$  and condition II.(b) in proposition \ref{cor:schur} is also satisfied. Hence the equilibrium is stable.
\end{proof}
 
%
%

\section{Existence of Equilibria and Stability Map in Two Bus Systems}\label{sec:example}

The analysis of the Jacobian reveals the stability of a given equilibrium and different routes to instability. We now investigate when a stable equilibrium exists and derive the global stability map of the grid. We focus on an elementary network of two coupled machines, one with positive effective power $P_1>0$ and one with negative effective power $P_2 = -P_1$. All other machine parameters are assumed to be identical,
 $E_2^{f} = E_1^{f}>0$, $X_2-X'_2=X_1-X'_1 > 0$ , $B_{2,2} = B_{1,1} <0$ and $B_{2,1}=B_{1,2} > 0$. 

This example allows for a fully analytical solution of the equilibrium equations (\ref{eq:3rd-fixed}). Stability requires that phase difference satisfies $\delta_1 - \delta_2 \in (-\pi/2, \pi/2)$. Equation \eqref{eq:3rd-fixed_b} in (\ref{eq:3rd-fixed}) can then be solved for the phase difference
\be
   \delta_{1}^\circ - \delta_{2}^\circ = \arcsin \left(\frac{P_1}{B_{1,2} E_1^\circ E_2^\circ} \right).
   \label{eq:2machine-d12}
\ee
Substracting  \eqref{eq:3rd-fixed_c} for one machine from that for the other machine one can show that the voltages at both machines must be identical in a stable equilibrium, 
\[
    E_2^\circ = E_1^\circ.
\]
Using the result (\ref{eq:2machine-d12}), we can now eliminate the phases from  \eqref{eq:3rd-fixed_c}. We are left with an equation containing only the state variable $E_1^\circ$:
\begin{align}
    & \Big\{  \big[ (X-X') B_{1,1} -1 \big]^2 - (X-X')^2 B_{1,2}^2 \Big \} E_1^{\circ 4}
          + E^{f 2}  E_1^{\circ 2} \nn \\
    &     + 2 \big[ (X-X') B_{1,1} -1 \big] E^{f}  E_1^{\circ 3} + (X-X')^2 P_1^2 = 0.   
    \label{eq:polynomial}
\end{align}
This is a fourth-order polynominal in $E_1^\circ$, which can be solved analytically leading to rather lengthy expressions. More importantly, the fundamental theorem of algebra tells us that there are exactly four solutions for a given set of system parameters. Discarding solutions where $E_{1,2}^\circ$ or $\delta_{1}^\circ - \delta_2^\circ$ are complex or $E_{1,2}^\circ < 0$, which are physically not meaningful, we obtain a set of equilibria, which can be stable or unstable, though. 

Solving equation (\ref{eq:polynomial}) as a function of the system parameters provides a general stability map and bifurcation set. Here, we focus on the dependence on the machine parameters $P_1=-P_2$ and $X-X'$, while keeping the transmission system parameters fixed. The stability map in Fig.~\ref{fig:stab-map} reveals three qualitatively different regimes. In region I the polynomial (\ref{eq:polynomial}) has two real roots, one corresponding to a stable and one to an unstable equilibrium. Hence a stable operation of the grid is possible. 

The phase difference decreases steadily with the increase of $X-X'$, whereas the nodal voltages increase. Notably, the maximum transmitted power (the maximum power $P_1$ for which a stable solution exists, corresponding to the boundary of regions I and II) firstly decreases with $X-X'$ and then increases for larger values of $X-X'$ due to 2nd power dependence on the voltages on nodes. Near the border of the two regimes I-III the maximum transmitted power is effectively infinite, as the voltage diverges.

\begin{figure}[tb]
\centering
\includegraphics[width=0.8\columnwidth]{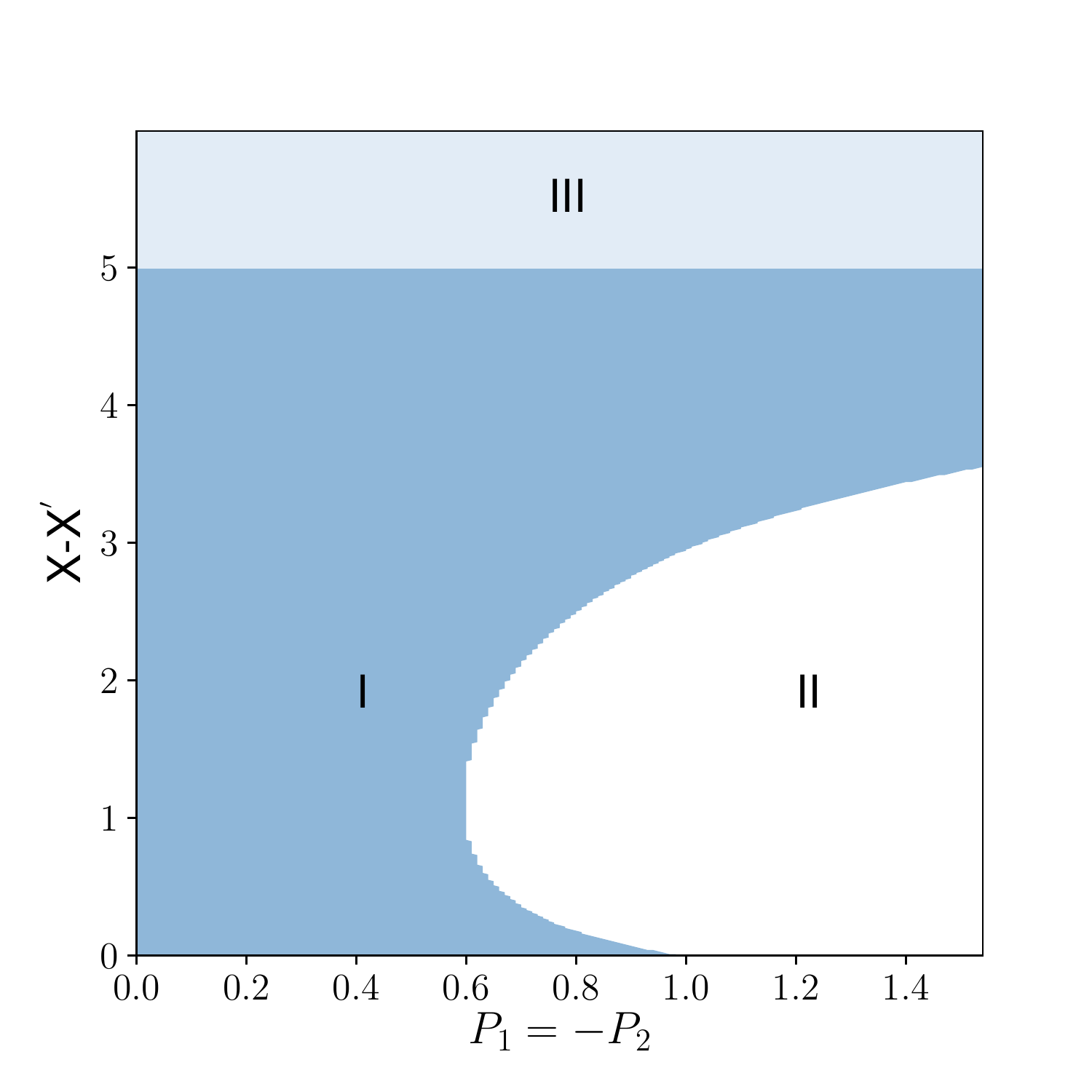}
\caption{
\label{fig:stab-map}
Stability map for a system of two coupled identical machines. A stable equilibrium exists only in region I. Stability is lost at the boundary to region II in an inverse saddle node bifurcation (see Figs.~\ref{fig:ins-freq} and \ref{fig:bifI-III} for examples). At the boundary to region III the voltage $E_1^\circ = E_2^\circ$ diverges and changes sign such that the equilibrium becomes unphysical. This corresponds to a pure voltage instability shown in Fig.~\ref{fig:bifI-II}.
The stability map is obtained from the solution of the fourth order polynomial (\ref{eq:polynomial}). At the boundary I-II two real solutions become complex and the discriminant/determinant of the polynomial changes sign.
}
\end{figure}

When the system is approaching the ``phase border'' between region I and II in the parameter space, two equilibria, stable and unstable, merge in an inverse saddle-node bifurcation, and the equilibrium is lost, accompanied with the birth of 2 complex roots of the polynomial (\ref{eq:polynomial}). The border line can be obtained analytically from the determinant of the fourth-order polynomial (\ref{eq:polynomial}) as it sign changes from $-$ to $+$ when passing through the line I-II. Therefore in region II all four roots of (\ref{eq:polynomial}) are complex and no stable equilibrium exits. For $X-X' = 0$ the saddle-node bifurcation corresponds to a pure rotor-angle instability (cf.~Fig.~\ref{fig:ins-freq}), for other values of $X-X'$ it corresponds to a mixed instability, cf. Fig.~\ref{fig:bifI-II}.

In region III the polynomial (\ref{eq:polynomial}) has two complex and two real roots with opposite signs. Only the positive real root corresponds to a physical solution, but it is always unstable, such that no stable operation is possible. By approaching the border I-III from below the coefficient of the fourth power in the polynomial (\ref{eq:polynomial}) becomes infinitesimally small. Consequently the magnitude of one root goes to infinity on the border, thereby changing its sign from positive in region I to negative in region III (cf.~Fig.~\ref{fig:bifI-III}). 

In physical terms, the voltage magnitude increases with $X-X'$ until it diverges on the border of the regions I-III. On the border one faces a pure voltage instability as illustrated in Fig.~\ref{fig:bifI-III}. No stable operation is possible if the difference of static and transient reactance $X-X'$ exceeds a critical value (region III). Notably the divergence of the voltage magnitude implies that the phase difference $\delta_1^\circ - \delta_2^\circ \rightarrow 0$ on the I-III-border if the transmitted power is finite. 

\section{Conclusion and Outlook} \label{sec:conclusions}
This paper has investigated the third-order model of power system dynamics with a focus on the interplay of rotor-angle and voltage stability. We employed a linearization approach and derived necessary and sufficient conditions for local exponential stability and uncovered different routes to instability. For the simplified case of a two-bus system, the positions of the equilibria as well as a global stability map was obtained fully analytically.

Our paper provides several rigorous results which can deepen our understanding into the factors limiting power system stability. In particular, the Proposition \ref{cor:schur} provides a decomposition of the Jacobian into the angle and voltage subspace by means of the Schur decomposition and thus allows to rigorously classify possible routes to instability. Stability of the angle subsystem requires that voltage phase angle differences remain small, see Corollaries \ref{cor:laplacian-pos} and \ref{cor:laplacian-pos2}. Stability of the voltage subsystem is threatened if the difference of the static and transient reactances $X-X'$ becomes too large as expressed in corollary \ref{cor:voltage1} and \ref{cor:voltage2}. Mixed instabilities can emerge due to the interplay of both subsystems.

Furthermore, our paper provides analytical insights into how the structure of a network determines its dynamics by linking stability properties to measures of connectivity. Future research can deepen the understanding and the applicability of our result. In particular, our approach can be extended towards models of increasing complexity such as the fourth order model. Analytic results can be compared to numerical simulations to test how tight the derived bounds are and to clarify the importance of line losses for the stability problem.

\begin{acknowledgments}

All authors gratefully acknowledge support from the Helmholtz Association via the joint initiative ``Energy System 2050 -- A Contribution of the Research Field Energy''. KS, LRG, MM, DW also acknowledge support by the Helmholtz Association under grant no.~VH-NG-1025 to DW and by the Federal Ministry of Education and Research (BMBF grant nos.~03SF0472).  TF acknowledges further support from the Baden-W\"urttemberg Stiftung in the Elite Programme for Postdocs. We thank Benjamin Sch\"afer, Moritz Th\"umler, Kathrin Schmietendorf, Oliver Kamps and Robin Delabays for stimulating discussions.

\end{acknowledgments}

\appendix*
\section{Construction of the Schur decomposition}

In the appendix we summarize some technical results which are used to proof the main results of the paper.

\begin{lemma}
\label{lem_for_theorem}
Consider the matrix $\vec U$ defined in (\ref{eq:defU}). If $\vec x \in \tilde{\mathcal{D}}_\perp$, then also $\vec y = \vec U \vec x \in \tilde{\mathcal{D}}_\perp$.
\end{lemma}
\begin{proof}
Let's write $\vec x = (\vec \xi \vec \epsilon)^\top$, then 
\begin{align}
	& \vec y = \begin{pmatrix} \eye & -\vec \Lambda^{+} \vec{A^\top}  \\ 0 & \eye \end{pmatrix} \begin{pmatrix}
	\vec \xi \\ \vec \epsilon
	\end{pmatrix} 
	 = \begin{pmatrix} \vec \xi - \vec \Lambda^{+} \vec{A^\top} \vec \epsilon \\ \vec \epsilon \end{pmatrix}.
\end{align}
$\eye \vec \xi = 0 $ because $\vec x \in \tilde{\mathcal{D}}_\perp$. $\eye \vec \Lambda^{+} \vec{A^{\top}}\vec \epsilon = 0$ because $\vec \Lambda^{+}$ maps onto $\bar{\mathcal{D}}_\perp$ as $\vec \Lambda$ maps onto $\bar{\mathcal{D}}_\perp$. Then $(\eye\vec \xi - \vec \Lambda^{+} \vec{A^\top} \vec \epsilon) = 0 $ and therefore $\vec y \in \tilde{\mathcal{D}}_\perp$. Also $\vec y \neq 0$, as if it is equal to zero, then 
\begin{align}
\begin{cases}
& \vec \xi - \vec \Lambda^{+} \vec{A^\top} \vec \epsilon = 0, \\ 
& \vec \epsilon = 0,
\end{cases}
\end{align}
which is realized if and only if
\begin{align}
\begin{cases}
& \vec \xi = 0, \\ 
& \vec \epsilon = 0,
\end{cases}
\end{align}
meaning $\vec x = 0$, that contradicts with the condition $\vec x \neq 0$.
\end{proof}



\bibliography{powerdyn_1}

\end{document}